\documentclass[10pt,conference]{IEEEtran} 
\IEEEoverridecommandlockouts

\usepackage[utf8]{inputenc} 
\usepackage[T1]{fontenc}
\usepackage{amsfonts}
\usepackage[cmex10]{amsmath} 
\usepackage{amsmath,amssymb,mathtools}
\usepackage{amsthm}
\usepackage{bbm}
\usepackage{bm}
\usepackage{times}
\usepackage{algorithm}
\usepackage{algpseudocode}
\usepackage{graphicx}
\usepackage{subcaption}
\usepackage{placeins}
\usepackage{xcolor}
\usepackage{booktabs}
\usepackage{enumitem}
\usepackage{framed}
\usepackage[hidelinks]{hyperref}

\newcommand{\R}{\mathbb{R}}

\newcommand{\bA}{\bm{A}}
\newcommand{\ba}{\bm{a}}
\newcommand{\bC}{\bm{C}}
\newcommand{\bb}{\bm{b}}
\newcommand{\tlock}{T_{\mathrm{lock\textrm{-}in}}}
\newcommand{\bc}{\bm{c}}
\newcommand{\bd}{\bm{d}}
\newcommand{\bI}{\bm{I}}
\newcommand{\bx}{\bm{x}}
\newcommand{\bxstar}{\bm{x}^{\star}}
\newcommand{\by}{\bm{y}}
\newcommand{\bu}{\bm{u}}
\newcommand{\bv}{\bm{v}}
\newcommand{\bU}{\bm{U}}
\newcommand{\bV}{\bm{V}}
\newcommand{\bR}{\bm{R}}
\newcommand{\br}{\bm{r}}
\newcommand{\bLambda}{\bm{\Lambda}}
\newcommand{\bDelta}{\bm{\Delta}}
\newcommand{\bg}{\bm{g}}

\newcommand{\be}{\bm{e}}
\newcommand{\bz}{\bm{z}}

\newcommand{\bh}{\bm{h}}
\newcommand{\bB}{\bm{B}}

\newcommand{\bQ}{\bm{Q}}
\newcommand{\bPhi}{\bm{\Phi}}
\newcommand{\bPi}{\bm{\Pi}}
\newcommand{\xhat}{\widehat{\bx}}
\newcommand{\betahat}{\widehat{\beta}}

\newcommand{\muhat}{\widehat{\mu}}
\newcommand{\abs}[1]{\left|#1\right|}
\newcommand{\trace}[1]{\textrm{Tr}\left[#1\right]}

\newcommand{\prob}{\mathbb{P}}

\DeclareMathOperator{\supp}{supp}

\DeclareMathOperator{\argmin}{arg\,min}

\newcommand{\norm}[1]{\left\|#1\right\|_2}
\newcommand{\opnorm}[2]{\left\|#1\right\|_{#2}}

\newcommand{\dualproj}[1]{\bm{P}^\perp_{#1}}
\newcommand{\algoname}{\texttt{CosMIHT}}

\theoremstyle{plain}
\newtheorem{theorem}{Theorem}
\newtheorem{lemma}{Lemma}
\newtheorem{proposition}{Proposition}

\theoremstyle{definition}

\newtheorem{assumption}{Assumption}
\theoremstyle{remark}
\newtheorem{remark}{Remark}

\title{Consecutive Support Matching Induced Parameter Tuning Accelerates Momentum Iterative Hard Thresholding}

\author{
Samrat Mukhopadhyay\IEEEauthorrefmark{1}, Member (IEEE) and Debasmita Mukherjee\IEEEauthorrefmark{2}
\thanks{
\IEEEauthorrefmark{1} is with the Department of Electronics Engineering, IIT (ISM) Dhanbad 
\IEEEauthorrefmark{2} is with the Department of Computer Science, Techno Bengal Institute of Technology 
}}

\begin{document}
\maketitle

\begin{abstract}
Momentum-based acceleration of iterative hard thresholding (IHT) can
dramatically speed up sparse signal recovery from linear measurements, but its effectiveness hinges
on careful parameter tuning---a task complicated by the frequent support
changes inherent to hard thresholding. We propose \algoname (Consecutive
Support Matching Induced Momentum IHT), which resolves this difficulty
through a simple adaptive rule: start with the conservative parameters and whenever two consecutive iterates share the
same support, estimate the extreme eigenvalues of the
support restricted Gram matrix via a lightweight power method and switch to the corresponding 
optimal heavy-ball parameters. This mechanism allows
\algoname\ to automatically interpolate between cautious MIHT-like behavior
during support discovery and near-optimal accelerated convergence after
support identification.
 
Under standard restricted isometry assumptions, we develop a two-phase
convergence theory. In the \emph{wandering phase}, we establish a
linear contraction of the recovery error up to a noise floor and derive an
explicit upper bound on the number of iterations required to identify the
correct support. In the \emph{lock-in} phase, we establish that, with a randomly initialized power method based eigenvalue estimates that depend on the number of power iterations, the algorithm enjoys, with high probability, a near-optimal accelerated convergence rate akin to the heavy ball method. We corroborate the theoretical findings with extensive numerical experiments on both noiseless and noisy measurements demonstrating that \algoname\ achieves faster convergence than state-of-the-art iterative sparse recovery techniques without compromising the recovery performance.
\end{abstract}

\begin{IEEEkeywords}
Iterative hard thresholding; consecutive support matching; heavy-ball; Restricted Isometry Property (RIP); support identification; compressed sensing.
\end{IEEEkeywords}

\section{Introduction}

We consider the sparse recovery problem: given a measurement vector $\by = \bA\bx^\star + \be \in \mathbb{R}^M$, where $\bA \in \mathbb{R}^{M \times N}$ is a known sensing matrix and $\be$ is measurement noise, recover the $K$-sparse vector $\mathbf{x}^\star \in \mathbb{R}^N$. This amounts to solving
\begin{equation}\label{eq:sparse_opt}
    \min_{\bx \in \mathbb{R}^N} \;\tfrac{1}{2}\|\by - \bA\bx\|_2^2 \quad \text{s.t.} \quad \|\bx\|_0 \leq K.
\end{equation}
Although \eqref{eq:sparse_opt} is NP-hard in general, it can be solved efficiently under regularity conditions on $\bA$, notably the restricted isometry property (RIP)~\cite{CandesTao2005}. Among the many algorithmic approaches, including convex relaxations such as Basis Pursuit~\cite{ChenDonohoSaunders1998} and LASSO~\cite{Tibshirani1996}, and greedy methods such as OMP~\cite{TroppGilbert2007}, CoSaMP~\cite{NeedellTropp2009}, and SP~\cite{DaiMilenkovic2009}, iterative hard thresholding (IHT)~\cite{BlumensathDavies2009} stands out for its simplicity: each iteration requires only a matrix-vector product and a hard thresholding step, yet IHT achieves near-optimal recovery guarantees under RIP.

Since IHT is essentially projected gradient descent onto the set of $K$-sparse vectors, a natural question is whether it can benefit from momentum-based acceleration. In smooth convex optimization over convex sets, Polyak's heavy-ball method~\cite{Polyak1964} and Nesterov's accelerated gradient~\cite{Nesterov1983} are known to achieve optimal first-order convergence rates, substantially outperforming vanilla gradient descent. This has motivated several accelerated IHT variants, such as, Fast Lipschitz Iterative Hard Tresholding (FLIHT)~\cite{Cevher2011}, Heavy Ball Hard Thresholding (HBHT)~\cite{SunZhouZhaoMeng2023} and Momentum IHT (MIHT)~\cite{JinXie2025}. Among these, MIHT has demonstrated particularly strong empirical performance by incorporating an exponentially weighted sum of all past gradient directions as the descent direction, rather than only the current gradient direction, which makes it conceptually analogous to the heavy-ball method.


Despite its empirical success, MIHT leaves a critical question unaddressed: \emph{how should the momentum parameter $\beta$ and step size $\mu$ be tuned?} In classical heavy-ball method for minimizing convex quadratic functions over convex sets, the optimal parameters are explicit functions of the condition number of the Hessian of the objective restricted to the feasible set. For IHT-type methods, however, this reasoning breaks down because the effective Hessian of the objective function changes as the \emph{active support changes across iterations}. When the support set $S_t = \mathrm{supp}(\bx^t)$ shifts, the restricted Hessian $\bA_{S_t}^\top \bA_{S_t}$ changes abruptly, invalidating any parameter choice tuned to the previous support. Aggressive momentum that accelerates convergence on a stable support can cause divergence during support transitions.

This tension between acceleration and support stability is the central difficulty in parameter selection for MIHT. Ideally, one wants cautious, safe parameters while the support is still being discovered, and aggressive, optimally tuned parameters once the support has stabilized. The key observation underlying this work is that \emph{consecutive support agreement} provides a cheap, reliable signal for deciding when to switch from conservative to aggressive parameters.


We propose \textbf{\algoname}\ (Consecutive Support Matching Induced Momentum IHT), which integrates an adaptive parameter tuning rule into MIHT. The core mechanism is simple: the algorithm starts with conservative parameter values and keeps using each iteration until  two consecutive iterates share the same support, distinct from the one before and then \algoname\ estimates the extreme eigenvalues of the restricted Gram matrix $\bA_{S_t}^\top \bA_{S_t}$ via a lightweight power method and switches to the optimal heavy-ball parameters using these eigenvalue estimates. Following this, \algoname\ keeps using these updated parameters until again two consecutive supports match, albeit, distinct from the one before. This allows the algorithm to interpolate automatically between safe MIHT-like behavior during support discovery and near-optimal accelerated convergence after support identification.

Our specific contributions are as follows:
\begin{enumerate}[label=(\roman*)]
    \item \textbf{Algorithm design.} We propose CoSMIHT (Algorithm~\ref{alg:CoSMIHT}), an IHT variant where the momentum and step-size parameters are updated adaptively based on consecutive support matching. The only hyperparameters requiring manual selection are the conservative defaults $(\mu, \beta)$ and the number of power iterations $r$.

    \item \textbf{Two-phase convergence theory.} We establish a rigorous two-phase convergence analysis under RIP:
    \begin{itemize}
        \item \emph{Wandering phase} (Propositions~\ref{prop:switched-evolution-inequalities}--\ref{prop:convergence-speed}): We prove global linear contraction of the error up to a noise floor, using conservative parameter bounds that hold uniformly across all support configurations.
        \item \emph{Lock-in phase} (Section~\ref{sec:lock-in}): Once the support stabilizes---which we show happens in finite time (Theorem~\ref{thm:lock-in-time})---the dynamics reduce to the classical heavy-ball iteration on a $K$-dimensional subspace, converging at a rate governed by $\sqrt{\hat{\beta}}$, where $\hat{\beta}$ is the optimally tuned momentum parameter.
    \end{itemize}

    \item \textbf{Finite lock-in time guarantees.} We provide an explicit upper bound on the number of iterations required for CoSMIHT to identify the correct support (Theorem~\ref{thm:lock-in-time}), expressed in terms of the RIP constant, the signal-to-noise ratio, and the minimum nonzero entry of~$\mathbf{x}^\star$.

    \item \textbf{Power method accuracy analysis.} We analyze the accuracy of the eigenvalue estimates produced by the power method (Proposition~\ref{prop:power_accuracy}), providing high-probability bounds on the number of power iterations $r$ needed to ensure that the tuned parameters are $\epsilon$-close to optimal.

    \item \textbf{Numerical validation.} We corroborate the theoretical findings with simulations comparing CoSMIHT against IHT, FLIHT, and MIHT in both noiseless and noisy settings, demonstrating improved recovery probability and faster convergence.
\end{enumerate}

\subsection{Paper Organization}
Section~\ref{sec:prelim} collects notation and preliminary results. Section~\ref{sec:algo} presents the CoSMIHT algorithm. Section~\ref{sec:convergence-analysis} contains the convergence analysis, including wandering phase bounds, lock-in time guarantees, and lock-in phase convergence rates as well as a discussion on the computational complexity. Section~\ref{sec:numerical} presents numerical experiments.

\section{Preliminaries}
\label{sec:prelim}
\paragraph*{Notation}
Bold lowercase/uppercase denote vectors/matrices.
Transpose is \((\cdot)^\top\).
For an index set \(S\subseteq[N]:=\{1,\dots,N\}\), \(\bx_S\) is the restriction, and \(\supp(\bx)\) the support.
\(\mathcal{H}_K(\bz)\) keeps the \(K\) largest-magnitude entries of \(\bz\) and zeros the rest.

\begin{assumption}[Restricted Isometry Property]
\label{ass:re}
The matrix $\bA\in \R^{M\times N}$ satisfies the restricted isometry property (RIP) of order $K$, with a restricted isometry constant (RIC) $\delta_K$ if the following holds:
\begin{align}
    \delta_K:=\max_{S\subset[N]:\abs{S}\le K}\opnorm{\bI_K-\bA_S^\top\bA_S}{2\to 2}<1,
\end{align}
where $\opnorm{\cdot}{2\to 2}$ denotes the operator norm.
\end{assumption}


\begin{lemma}[Theorem 1~\cite{shen2018tight}]
\label{lem:cone}
Let \(\bxstar\) be \(K\)-sparse with support $S$. For any \(\bz\),
\(
\|\mathcal{H}_K(\bz)-\bxstar\|_2 \le \omega\|(\bz-\bxstar)_{S\cup T}\|_2
\)
where $T=\mathrm{supp}(\mathcal{H}_K(\bz))$, and \(\omega =\frac{1+\sqrt{5}}{2}\) is the hard-thresholding expansion constant. 
\end{lemma}

\section{Proposed Algorithm}
\label{sec:algo}
\begin{algorithm}[t!]
    \caption{\texttt{POWER} Method}
    \label{alg:pow}
    \begin{algorithmic}[1]
        \Require $\bB\in \R^{K\times K}$, $\bb\in \R^K$, $r\in \mathbb{Z}_+$
        \State \textbf{Initialize} $\bb_0 = \bb$
        \For{$k = 1,2,\cdots, r$}
            \State $\bb_k = \bB\bb_{k-1}$
        \EndFor
        \State \textbf{Return} $\frac{\bb_{r}}{\norm{\bb_{r}}}$
    \end{algorithmic}
\end{algorithm}
\begin{algorithm}[t!]
\caption{\texttt{CoSMIHT}: Consecutive Support Matching Induced Momentum Iterative Hard Thresholding}
\label{alg:CoSMIHT}
\begin{algorithmic}[1]
\Require $\by,\bA$, sparsity $K$, stepsize $\mu>0$, momentum $\beta\in[0,1)$, no. of power iterations $r$, time horizon $T$
\State \textbf{Initialize:} $\bx^{0}=\mathbf{0}$, $\bv^{0}=\mathbf{0},\bb_0=\frac{\bg}{\norm{\bg}},$ where $\bg\sim\mathcal{N}(\bm{0},\bI)$,
$S_{-1}=\emptyset,$ $\beta_{-1}=\beta$, $\mu_{-1}=\mu$
\For{$t=0,1,2,\dots, T-1$}
    \State $S_t = \supp(\bx^t)$
    \Statex {\color{blue}\underline{\textbf{Consecutive support matching induced parameter tuning}}}
    \If{$t>0$, $S_t=S_{t-1}$ and $S_{t-1}\ne S_{t-2}$}
            \State $\bb_{t,r} = \textrm{POWER}(\bA_{S_t}^\top\bA_{S_t}, \bb_{0}, r)$
            \State $\bc_{t,r} = \textrm{POWER}( \opnorm{\bA_{S_t}}{F}^2\bI_K - \bA_{S_t}^\top\bA_{S_t}, \bb_{0}, r)$
            \State $l_t = \norm{\bA_{S_t}\bc_{t,r}}^2$, $L_t =  \norm{\bA_{S_t}\bb_{t,r}}^2$
            \State $\mu_t\gets \frac{4}{(\sqrt{L_t}+\sqrt{l_t})^2}$, $\beta_t\gets \left(\frac{\sqrt{L_t}-\sqrt{l_t}}{\sqrt{L_t}+\sqrt{l_t}}\right)^2$
    \Else 
        \State ${\beta}_t \gets \beta_{t-1}$, $\mu_t \gets \mu_{t-1}$
    \EndIf
    \State $\bg^t = -\bA^\top(\by-\bA\bx^t)$
    \State $\bv^{t+1} = {\beta}_t\bv^t - {\mu}_t\bg^t$ 
    \State $\bx^{t+1} \gets \mathcal{H}_K\!\big(\bx^t + \bv^{t+1}\big)$
\EndFor
\State\textbf{Output:} $\bm{x}^{T}$
\end{algorithmic}
\end{algorithm}
We propose \algoname\ (Algorithm~\ref{alg:CoSMIHT}) by integrating a \textit{consecutive support matching induced parameter tuning step} with the MIHT computations. Specifically, at step $t$, we first estimate the support $S_t$ of the current iterate $\bx^t$ and update the momentum parameter $\beta_t$ and step size $\mu_t$ based on $S_t, S_{t-1}$, and $S_{t-2}$, as follows. If there is a support change at step $t$, i.e., $S_{t-1}\ne S_t$, then the supports have not stabilized yet and we retain the previous parameters $\beta_t = \beta_{t-1}, \mu_t = \mu_{t-1}$. If $S_t = S_{t-1}$, then there are two cases to consider. If $S_{t-1}\ne S_{t-2}$, i.e., the support just stabilized at $t$, then we update the parameters $\beta_t$ and $\mu_t$ to the optimal parameters for Polyak's heavy-ball method (see lines~5--8 of Algorithm~\ref{alg:CoSMIHT}). If $S_{t-1} = S_{t-2}$, then the support stabilized before $t-1$, and \algoname\ reuses the parameters from step $t-1$.

\section{Theoretical Results}
\label{sec:convergence-analysis}
\subsection{Convergence analysis of \texttt{CoSMIHT}}
We begin this section with a convergence analysis of \texttt{CoSMIHT}. Recall that we aim to solve the optimization problem 
\begin{align}
\begin{aligned}
    \ & \min \frac{1}{2}\norm{\by-\bA\bx}^2\ \mathrm{s.t.}\ \opnorm{\bx}{0}\le K,\quad (P_1)
\end{aligned}
\end{align}
where we adopt the measurement model $\by=\bA\bx^\star+\be$ where $\bx^\star$ is $K$-sparse with support $S^\star$. 

We first establish that, under sufficient signal-to-noise ratio, the optimization problem $(P_1)$ has a unique $K$-sparse solution supported on $S^\star$, which serves as the target for our convergence analysis.
\begin{proposition}
    The problem $(P_1)$ admits a unique solution $\widehat{\bx}$ with support $S^\star$ if 
    \begin{align}
    \label{eq:unique-soln-condition}
        \frac{\norm{\bx^\star}}{\norm{\be}} & > \frac{2}{1-\delta_{2K}}.
    \end{align}
    Furthermore, under~\eqref{eq:unique-soln-condition}, the following holds:
    \begin{align}
        \label{eq:opt-soln}
        \widehat{\bx}_{S^\star} & =\bx^\star_{S^\star}+\bA^\dagger_{S^\star}\be.
    \end{align}
\end{proposition}
\begin{proof}
    Note that $S^\star$ is the support of one of the optimal solutions of $(P_1)$ if $S^\star\in \argmin_{S\subset [N]:\abs{S}=K}\norm{\dualproj{S}\by}^2$, which is equivalent to $\norm{\dualproj{S^\star}\by}\le \norm{\dualproj{S}\by},\forall S\subset[N]:\abs{S}=K$. Using the measurement model,
    \begin{align}
    \norm{\dualproj{S^\star}\by} & =\norm{\dualproj{S^\star}\be}\le \norm{\be},
    \end{align}
    Therefore,
    \begin{align}
        \norm{\dualproj{S}\by} & \ge \norm{\dualproj{S}\bA_{S^\star\setminus S}\bx^\star_{S^\star\setminus S}}-\norm{\dualproj{S}\be}\nonumber\\
        \ & \ge (1-\delta_{2K})\norm{\bx^\star}-\norm{\be}.
    \end{align}
    Consequently, $S^\star\in \argmin_{S\subset [N]:\abs{S}=K}\norm{\dualproj{S}\by}$ if 
    \begin{align}
        2\norm{\be} & \le (1-\delta_{2K})\norm{\bx^\star}.
    \end{align}
    Therefore, $S^\star$ is the unique optimal set if the condition~\eqref{eq:unique-soln-condition} is satisfied. Furthermore, the optimal solution $\widehat{\bx}$ to $(P_1)$ then becomes the least squares solution on the support $S^\star$, so that $\widehat{\bx}_{S^\star}=\bA_{S^\star}^\dagger\by$, which results in~\eqref{eq:opt-soln}.
\end{proof}

We divide the convergence analysis into two phases. The first phase is termed the \textbf{wandering phase}, where essentially the algorithm wanders through different supports, and the second phase is referred to as the \textbf{lock-in phase} where the algorithm locks in to a fixed support. To define these phases mathematically, we first define the \textit{lock-in} time as below:
\begin{align}
\label{eq:lock-in-def}
    T_0 & = \argmin\left\{t\ge 1:S_\tau = S^\star, \ \forall \tau\ge t\right\}.
\end{align}
Using this definition, the two phases are equivalent to the time intervals $I_i,\ i=1,2$, defined as below:
\begin{align}
    I_1 & = \left\{1\le t\le T_0-1\right\}, & 
    (\mbox{\textbf{Wandering Phase}})\\
    I_2 & = \left\{T_0\le t\le T\right\}. &  (\mbox{\textbf{Lock-in Phase}})
\end{align}
We now analyze the \texttt{CoSMIHT} algorithm for each of these phases.
\subsubsection{Wandering phase convergence analysis of \texttt{CoSMIHT}}
\label{sec:wandering-phase-analysis}
In this phase, as the support is not stabilized, \algoname\  undergoes frequent support switching. By combining the analysis of MIHT~\cite{JinXie2025} with the parameter selection rule of \algoname\ (Lines~5--10 in Algorithm~\ref{alg:CoSMIHT}), we state two results which analyze the evolution of $\bx^t$ and $\bv^t$ in this phase. Recall from Lemma~\ref{lem:cone} that $\omega \ge 1$ denotes the hard-thresholding expansion constant.
\begin{proposition}
    \label{prop:switched-evolution-inequalities}
    Let $E_t = \norm{\bx^t-\widehat{\bx}}$ and $V_t=\norm{\mathcal{H}_{2K}(\bv^t)}$, for each $t\ge 1$. Then, the following holds:
    \begin{align}
    \label{eq:CoSMIHT-evolution-base}
        \begin{bmatrix}
            E_{t} \\
            V_{t}
        \end{bmatrix} & \preceq \bQ \begin{bmatrix}
            E_{t-1} \\
            V_{t-1}
        \end{bmatrix} + \xi \norm{\be}\begin{bmatrix}
           1\\
           1
        \end{bmatrix},
    \end{align}
    where 
    \begin{align}
    \label{eq:P-mat-expression}
        \bQ = \begin{bmatrix}
            \omega \rho & \omega \alpha\\
            \gamma\sqrt{1+\delta_{3K}} & \alpha
        \end{bmatrix}, 
    \end{align}
    with 
    \begin{align}
    \label{eq:gamma-alpha}
        \ & \gamma = \max\left\{\frac{1}{1-\delta_{3K}}, \mu\right\},
        \ \alpha = \max\left\{\delta_{3K}^2, \beta\right\},\\
        \label{eq:rho-xi}
        \ & {\rho} = \max\left\{\abs{\mu-1}+\mu\delta_{3K},\frac{2\delta_{3K}}{1-\delta_{3K}}\right\},
        \ \xi = \omega \gamma\sqrt{1+\delta_{3K}}.
    \end{align}
\end{proposition}
\begin{proof}
    The proof is postponed to Appendix~\ref{appendix:proof-prop2}.
\end{proof}
Using Proposition~\ref{prop:switched-evolution-inequalities}, we now obtain the following result.
\begin{proposition}
    \label{prop:convergence-speed}
    Under the conditions
    \begin{align}
    \label{eq:delta-alpha-mu-condition-convergence}
        \delta_{3K}<\frac{1}{2\omega+1},\ \frac{\omega\gamma\sqrt{1+\delta_{3K}}}{1-\omega\rho} < \frac{1-\alpha}{\alpha},\ \mu<\frac{1}{1-\delta_{3K}},
    \end{align}
    \texttt{CoSMIHT} satisfies
    \begin{align}
        E_t & \le \frac{\norm{\widehat{\bx}}}{2}(\lambda_1^t(1-\sin\theta)+\lambda_2^t(1+\sin\theta))\nonumber\\
        \label{eq:Et-wandering-equation}
        \ & + \frac{\xi\norm{\be}((\omega-1)\alpha+1)}{(1-\lambda_1)(1-\lambda_2)},\\
        V_t & \le \frac{\norm{\widehat{\bx}}b\cos\theta}{2\omega \alpha}(\lambda_1^t-\lambda_2^t)\nonumber\\
        \label{eq:Vt-wandering-equation}
        \ & + \frac{\xi\norm{\be}\cos\theta}{(1-\lambda_1)(1-\lambda_2)}\left[\gamma\sqrt{1+\delta_{3K}} - \omega\rho + \frac{\alpha+1}{2}\right],\\
        \label{eq:eig-condition}
        \ & \lambda_2<0<\lambda_1<1,
    \end{align}
    where $\tan\theta=\frac{\alpha-\omega\rho}{2b}$, $b=\sqrt{\omega\alpha\gamma\sqrt{1+\delta_{3K}}}$ and \begin{align}
        \lambda_{i} =\omega\rho+b(\tan\theta+(-1)^{i-1}\sec\theta),\ i=1,2.
    \end{align}
\end{proposition}
\begin{proof}
   The proof is postponed to Appendix~\ref{appendix:proof-prop3}.
\end{proof}
    
\subsubsection{Lock-in Time analysis of \algoname}
\label{sec:lock-in-time-analysis}

We now show that if the noise is sufficiently small, i.e., if $\norm{\be}$ is small enough, then \algoname\ always locks in to the correct support $S^\star$ after a finite number of iterations. 
\begin{theorem}
    \label{thm:lock-in-time}
    Under the condition~\eqref{eq:delta-alpha-mu-condition-convergence}, the lock-in time is upper bounded as below:
    \begin{align}
    \label{eq:lock-int-ime-upper-bound}
        \tlock\le \left\lceil\frac{\ln\left(\frac{\sqrt{2}\norm{\bx^\star}\left(\rho+\frac{b\cos\theta}{\omega}\right)}{x^\star_{\min}-\xi\norm{\be}\upsilon}\right)}{\ln(1/\lambda_1)}\right\rceil,
    \end{align}
    where 
    \begin{align}
        \upsilon & = \frac{\rho}{(1-\lambda_1)(1-\lambda_2)}\Bigg[(\omega-1)\alpha+1\nonumber\\
        \ & +\alpha\cos\theta\left(\gamma\sqrt{1+\delta_{3K}}-\omega\rho+\frac{\alpha+1}{2}\right)\Bigg].
    \end{align}
\end{theorem}
\begin{proof}
    We need to find a $t$ such that $S_{t+1}=S^\star$ for all $\tau\ge t+1$. For this to happen, we must ensure that, at time $t$, for any $i\in S^\star$ and $j\notin S^\star$, the following holds:
    \begin{align}
    \label{eq:lock-in-condition}
        \abs{(\bm{x}^t+\bm{v}^{t+1})_i} & \ge \abs{(\bm{x}^t+\bm{v}^{t+1})_j}.
    \end{align}
    Now, the left hand side (LHS) of~\eqref{eq:lock-in-condition} can be lower bounded as below:
    \begin{align}
        \lefteqn{\abs{(\bm{x}^t+\bm{v}^{t+1})_i} \ge \abs{\widehat{x}_i} - \abs{(\bm{x}^t-\widehat{\bx}+\bm{v}^{t+1})_i}} & &\nonumber\\
        \ & = \abs{\widehat{x}_i} - \abs{(\bm{x}^t-\widehat{\bx}+\beta_t\bm{v}^{t}-\mu_t\bg^t)_i}\nonumber\\
        \ & \ge \abs{\widehat{x}} - \abs{((\bI-\mu_t\bA^\top\bA)(\bm{x}^t-\widehat{\bx}))_i}\nonumber\\
        \ & -\beta_t\abs{v_i^{t}}-\mu_t\abs{(\bA^\top\be)_i}.
    \end{align}
    On the other hand, the right hand side (RHS) of~\eqref{eq:lock-in-condition} can be upper bounded as below:
    \begin{align}
        \lefteqn{\abs{(\bm{x}^t+\bm{v}^{t+1})_j} = \abs{(\bm{x}^t-\widehat{\bx}+\beta_t\bm{v}^{t}-\mu_t\bg^t)_j}} & &\nonumber\\
        \ & \le \abs{((\bI-\mu_t\bA^\top\bA)(\bx^t-\widehat{\bx}))_j} \nonumber\\
        \ & + \beta_t\abs{v^t_j}+\mu_t\abs{(\bA^\top\be)_j}
    \end{align}
    Therefore, the lock-in condition~\eqref{eq:lock-in-condition} is satisfied if the following holds:
    \begin{align}
        \lefteqn{\abs{((\bI-\mu_t\bA^\top\bA)(\bx^t-\widehat{\bx}))_j}  + \beta_t\abs{v^t_j}+\mu_t\abs{(\bA^\top\be)_j}} & &\nonumber\\
        \ & < \abs{\widehat{x}_i} - \abs{((\bI-\mu_t\bA^\top\bA)(\bm{x}^t-\widehat{\bx}))_i}\nonumber\\
        \ & -\beta_t\abs{v_i^{t}}-\mu_t\abs{(\bA^\top\be)_i}.
    \end{align}
    which in turn is satisfied if
    \begin{align}
        \sqrt{2}\left[\abs{((\bI-\mu_t\bA^\top\bA)(\bx^t-\widehat{\bx}))_{i,j}} + \beta_t\abs{(\bv^t)_{i,j}}\right] < \abs{\widehat{x}_i}.
    \end{align}
    Recalling the condition~\eqref{eq:opt-soln}, a sufficient condition ensuring the above is 
    \begin{align}
        \sqrt{2}(\abs{\mu_t-1}+\mu_t\delta_{3K}) E_t + \beta_tV_t\le x^\star_{\min},
    \end{align}
    where $x^\star_{\min} = \min_{i\in S^\star}\abs{x_i}.$
    Recalling the bounds of $\rho_t=\abs{\mu_t-1}+\mu_t\delta_{3K}$ and $\beta_t$ from Proposition~\ref{prop:switched-evolution-inequalities}, we get a sufficient condition for lock-in at time $t$ as below:
    \begin{align}
        \sqrt{2}(\rho E_t + \alpha V_t)\le x^\star_{\min}.
    \end{align}
    Recalling the evolution inequalities~\eqref{eq:Et-wandering-equation} and~\eqref{eq:Vt-wandering-equation}, which hold by virtue of the conditions~\eqref{eq:delta-alpha-mu-condition-convergence}, and recalling that eigenvalue of $\bQ$ satisfy $|\lambda_2|<\lambda_1$ from~\eqref{eq:eig-condition}, the above is satisfied if the following holds:
    \begin{align}
    \lefteqn{\sqrt{2}\norm{\bx^\star}\left(\rho+\frac{b\cos\theta}{\omega}\right)\lambda_1^t} & &\nonumber\\
        \ & <x^\star_{\min} - \frac{\xi\rho\norm{\be}}{(1-\lambda_1)(1-\lambda_2)}\Bigg[(\omega-1)\alpha+1\nonumber\\
        \ & +\alpha\cos\theta\left(\gamma\sqrt{1+\delta_{3K}}-\omega\rho+\frac{\alpha+1}{2}\right)\Bigg],
    \end{align}
    which in turn can be easily checked to be satisfied if $t\ge \tlock$ from~\eqref{eq:lock-int-ime-upper-bound}. 
\end{proof}
\subsubsection{Lock-in phase convergence analysis of \algoname}
\label{sec:lock-in}
Theorem~\ref{thm:lock-in-time} guarantees that \algoname\ identifies the correct support in finite time. 
Once \algoname\ locks in to the correct support $S^\star$ (which is guaranteed under sufficient signal-to-noise ratio as given in Theorem~\ref{thm:lock-in-time}), the parameters $\beta_t,\mu_t$ are fixed, and the trajectory of \algoname\ becomes identical to that of the heavy-ball method, with the Hessian being $\bA_{S^\star}^\top\bA_{S^\star}$. 

To conduct a convergence analysis of \algoname\ in the lock-in phase, we first look at the update steps constrained to $S^\star$. We first denote $\muhat = \mu_{\tlock}$ and $\betahat = \beta_{\tlock}$. Now note that for $t\ge \tlock$, the support of $\bx^t$ is $S^\star$, so that
\begin{align}
    \bx^{t+1}_{S^\star} = \bx^t_{S^\star}+\bv^{t+1}_{S^\star}, t\ge \tlock.
\end{align}
Then, using the fact that $\dualproj{S^\star}\bA_{S^\star}=\bm{0}_{S^\star}$, for $t\ge \tlock+1$, 
\begin{align}
    \lefteqn{\bx^{t+1}_{S^\star} = \bx^t_{S^\star} + \bv^{t+1}_{S^\star}} & & \nonumber\\
    \ & = \bx^t_{S^\star} + \betahat\bv^t_{S^\star}+\muhat\bA^\top_{S^\star}(\by-\bA_{S^\star}\bx^t_{S^\star}) \nonumber\\
    \ & \implies \bx^{t+1}_{S^\star}-\xhat_{S^\star} = (\bI_K-\muhat\bA^\top_{S^\star}\bA_{S^\star})(\bx^t_{S^\star} - \xhat_{S^\star})\nonumber\\
    \ & + \betahat(\bx^t_{S^\star}-\bx^{t-1}_{S^\star})\nonumber\\
    \ & \implies \bh^{t+1} = (\bI_K-\muhat\bPhi)\bh^t+\betahat(\bh^t-\bh^{t-1}),
\end{align}
where $\bh^t=\bx^t_{S^\star}-\xhat_{S^\star}$ and $\bPhi = \bA_{S^\star}^\top\bA_{S^\star}$. This is the update equation of the classical Polyak's heavy-ball method. Although analysis of this system is classical, we present here for completeness and then elaborate on the implications.
\begin{proposition}
    \label{prop:heavy-ball-convergence}
    The sequence $\{\bh^t\}_{t\ge \tlock+1}$ enjoys the following convergence guarantee:
    \begin{align}
        \label{eq:heavy-ball-convergence}
        \norm{\bh^t} & = \mathcal{O}\left(\sqrt{K}\widetilde{\beta}^{\frac{t-\tlock-1}{2}}\right),
    \end{align}
    where $g(x) = x+\sqrt{x^2-1}$ and 
    \begin{align}
        \label{eq:beta-tilde}
        \sqrt{\widetilde{\beta}} & =  \sqrt{\betahat}g\left(\max\left\{\max_{1\le i\le K}\abs{\frac{1+\betahat-\muhat\lambda^\star_i}{2\sqrt{\betahat}}},1\right\}\right),
    \end{align}
    where $\lambda^\star_1\ge\cdots\ge \lambda^\star_K$ are the eigenvalues of $\bA^\top_{S^\star}\bA_{S^\star}$.
\end{proposition}
\begin{proof}
    The proof is postponed to Appendix~\ref{appendix:proof-prop4}.
\end{proof}
 Now recall that post lock-in we have (line $8$ of Algorithm~\ref{alg:CoSMIHT}) \begin{align}
    \betahat = \left(\frac{\sqrt{L}-\sqrt{l}}{\sqrt{L}+\sqrt{l}}\right)^2,\  & \muhat = \frac{4}{(\sqrt{L}+\sqrt{l})^2},
\end{align}
where 
\begin{align}
\ & \begin{aligned}
    L & = \norm{\bA_{S^\star}\bb}^2,\\ l & = \norm{\bA_{S^\star}\bc}^2,
\end{aligned}\\
\ & \begin{aligned}
    \bb & = \textrm{POWER}(\bA_{S^\star}^\top\bA_{S^\star}, \bb_0, r),\\
    \bc & = \textrm{POWER}(\opnorm{\bA_{S^\star}}{F}^2\bI_K-\bA_{S^\star}^\top\bA_{S^\star}, \bb_0, r).
\end{aligned}
\end{align}
Let $\bB=\bA_{S^\star}^\top\bA_{S^\star}$, and denote the eigenvalues of $\bB$ by
$\lambda^\star_1\ge\cdots\ge \lambda^\star_K$, which satisfy
$1-\delta_K\le\lambda^\star_K\le\lambda^\star_1\le 1+\delta_K$ by the RIP of $\bA$.
When $L=\lambda_1$ and $l=\lambda_K$, the parameters $\betahat$ and $\muhat$ reduce
to the optimal heavy-ball parameters for $\bB$. This guarantee is achievable in the
limit of indefinitely many power iterations; however, since only finitely many
iterations are performed in practice, there is an approximation error between $L$
and $\lambda_1$, and between $l$ and $\lambda_K$. This error renders $\betahat$ and
$\muhat$ suboptimal and consequently degrades the convergence rate $\sqrt{\widetilde{\beta}}$.
We now show that, with an appropriate initialization of the power method, a certain
number of iterations suffices to control this approximation error within a prescribed
tolerance, thereby bringing the resulting convergence rate within a desired accuracy
of the optimal rate $ \left(\frac{\sqrt{\lambda^\star_1}-\sqrt{\lambda^\star_K}}{\sqrt{\lambda^\star_1}+\sqrt{\lambda^\star_K}}\right)^2.$
\begin{proposition}
\label{prop:power_accuracy}
    Let there be a number $c>1$ such that the matrix $\bA^\top_{S^\star}\bA_{S^\star}$ satisfies $\max\{\mu_2,\eta_{K-1}\}\le \frac{1}{c}$, where $\mu_2=\frac{\lambda_2^\star}{\lambda_1^\star}$ and $\eta_{K-1}=\frac{\sum_{i\ne K-1}\lambda_i^\star}{\sum_{i\ne K}\lambda^\star_i}$. Then, with the initialization of Algorithm~\texttt{POWER} used in \algoname\ , and with $K\ge 3$, $\epsilon>0$, the following is ensured with probability greater than $1-1/K^p$, for any $p\ge 1$,
    \begin{align}
    \begin{aligned}
        r\ge \left\lceil\frac{\ln\left(\frac{4K^{p+2}\lambda^\star_1}{\pi\epsilon}\right)}{2\ln c}\right\rceil
    \end{aligned}
        \implies 
        \begin{aligned}
            \lambda^\star_1-\epsilon\le L & \le \lambda^\star_1,\\
            \lambda^\star_K\le l & \le \lambda_K^\star+\epsilon.
        \end{aligned}
    \end{align}
\end{proposition}
\begin{remark}[Practical rule of thumb]
Since $\max\{\mu_2, \eta_{K-1}\} < 1$ and the logarithmic dependence on $K$ and $\epsilon^{-1}$ is mild, the number of power iterations required is typically $r = \mathcal{O}(\log K)$ for moderate accuracy, making the eigenvalue estimation lightweight.
\end{remark}
\begin{proof}
    The proof is postponed to Appendix~\ref{appendix-proof-prop5}.
\end{proof}
Using the above proposition, we come to the following result:
\begin{theorem}
    \label{thm:lock-in-convergence}
    If $\bA_{S^\star}$ satisfies the regularity conditions as in Proposition~\ref{prop:power_accuracy} and with $K\ge 3$, $\epsilon>0$ and $p\ge 1$, if $r$ is specified as in Proposition~\ref{prop:power_accuracy}  then for any $t\ge\tlock+1$, the following holds with probability $\ge 1-1/K^p$,
    \begin{align}
        \norm{\bx^t-\widehat{\bx}} & = \mathcal{O}\left(\sqrt{K}\left(\frac{\delta_{3K}\left(1+\frac{2\epsilon}{\lambda^\star_1-\lambda^\star_K}\right)}{1+\sqrt{1-\delta_{3K}^2}}\right)^t\right).
    \end{align}
\end{theorem}
\begin{proof}
    The proof is postponed to Appendix~\ref{appendix-proof-thm2}.
\end{proof}
\subsection{Computational Complexity Analysis}
\label{sec:computational-complexity}
In this section we carry out a computational complexity analysis of the proposed \algoname\ in terms of the number of floating point operations (flops) required to execute their steps. It follows from the description of \algoname\ in Algorithm~\ref{alg:CoSMIHT} that at step $t$, there are two types of computations involved: the compulsory computations (lines~11--13 of Algorithm~\ref{alg:CoSMIHT}) and the conditional computations (lines~5--8 of Algorithm~\ref{alg:CoSMIHT}). In the following, we list down the (order of) flops for different computational tasks.\\  
For the conditional computations,
\begin{enumerate}
    \item Power iterations for $L_t$ computation (line~5): $\mathcal{O}(rMK)$,
    \item Power iterations for $l_t$ computation (line~6): $\mathcal{O}(rMK)+\mathcal{O}(rK^2)$,
    \item $L_t,l_t$ computations (line~7): $\mathcal{O}(KM)+\mathcal{O}(K)$.
\end{enumerate}
For the compulsory computations,
\begin{enumerate}
    \item Gradient computation (line~11): $\mathcal{O}(MN)$,
    \item Momentum update (line~12): $\mathcal{O}(N)$,
    \item Estimate update (line~13): $\mathcal{O}(N)+\mathcal{O}(K\ln N)$.
\end{enumerate}
Note that the extra computational burden of \algoname\ compared to MIHT is only due to the conditional computations. The description of the support matching condition in Algorithm~\ref{alg:CoSMIHT} implies that the conditional computations are executed only when the current support matches with the previous support, which in turn differs from the support before. Consequently, the maximum number of times this can happen within the wandering phase is limited by $\tlock/2$. Furthermore, when \algoname\ enters the lock-in phase, it requires $1$ iteration of the conditional computation for the last time, totaling at most $\tlock/2+1$ iterations of conditional computations.  Therefore, the total conditional computation amounts to be $\mathcal{O}\left(rMK\tlock\right)$. Since $\tlock$ is essentially $\mathcal{O}\left(\frac{\norm{\bx^\star}}{x^\star_{\min}}\right)$, and since for $\bx^\star$ with subGaussian non-zero entries $\frac{\norm{\bx^\star}}{x^\star_{\min}}\approx \sqrt{K}$ with high probability, the extra conditional computations of \algoname\ is about $\mathcal{O}(rKM\ln K)$, which can be substantially smaller than the total compulsory computations of $\mathcal{O}(TMN)$. Therefore, \algoname\ is essentially computationally as light as MIHT, yet enjoying the improved convergence as depicted in the previous sections.  

\section{Numerical Experiments}
\label{sec:numerical}
\begin{figure}[t!]
    \centering
    \begin{subfigure}{0.24\textwidth}
    \centering
        \includegraphics[scale=0.28]{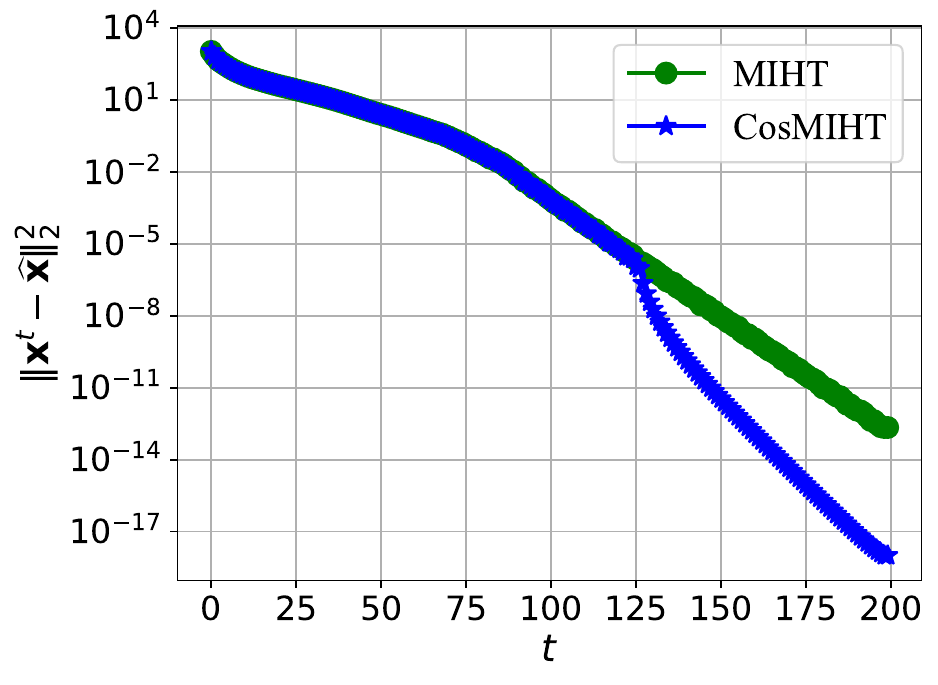}
        \subcaption{$\sigma=0.$}
    \end{subfigure}
    \begin{subfigure}{0.24\textwidth}
    \centering
        \includegraphics[scale=0.28]{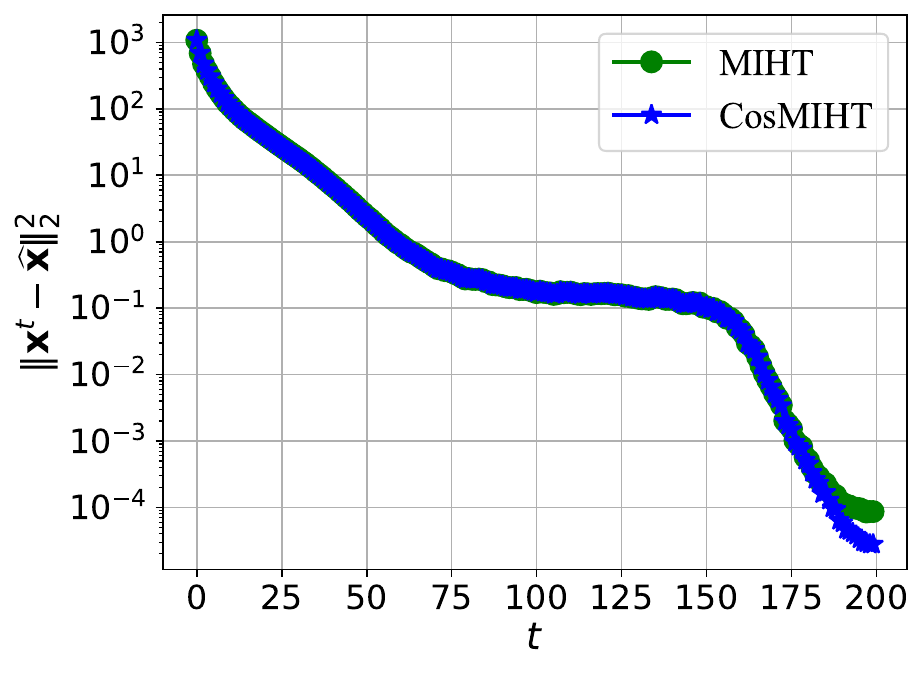}
        \subcaption{$\sigma=10^{-3}$}
    \end{subfigure}
    \caption{MSE vs iterations for noiseless and noisy measurements.}
    \label{fig:error-vs-it}
\end{figure}
\begin{figure}[t!]
    \centering
    \begin{subfigure}{0.24\textwidth}
    \centering
        \includegraphics[scale=0.3]{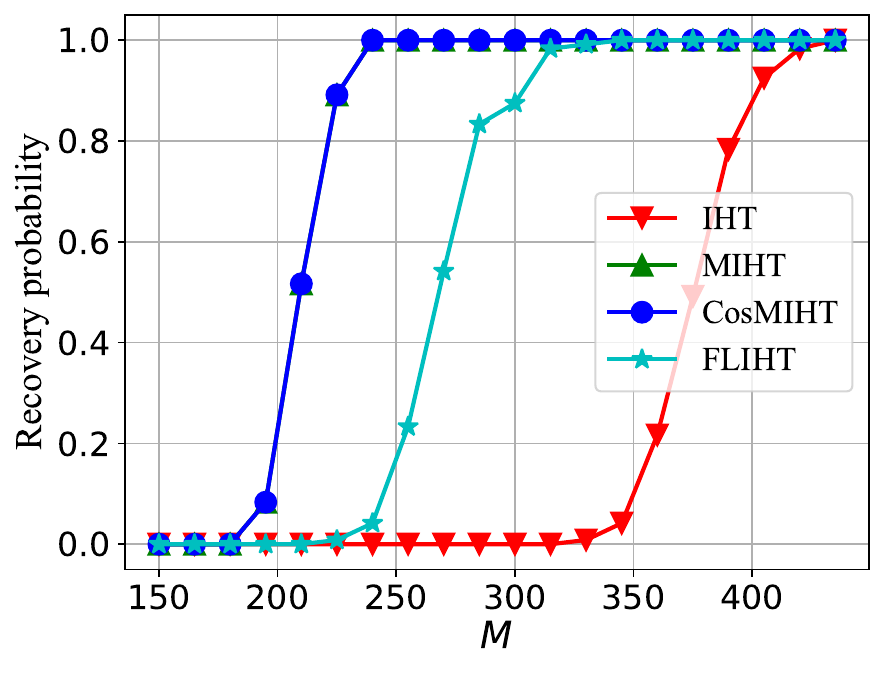}
        \subcaption{$\sigma=0.$}
    \end{subfigure}
    \begin{subfigure}{0.24\textwidth}
    \centering
        \includegraphics[scale=0.3]{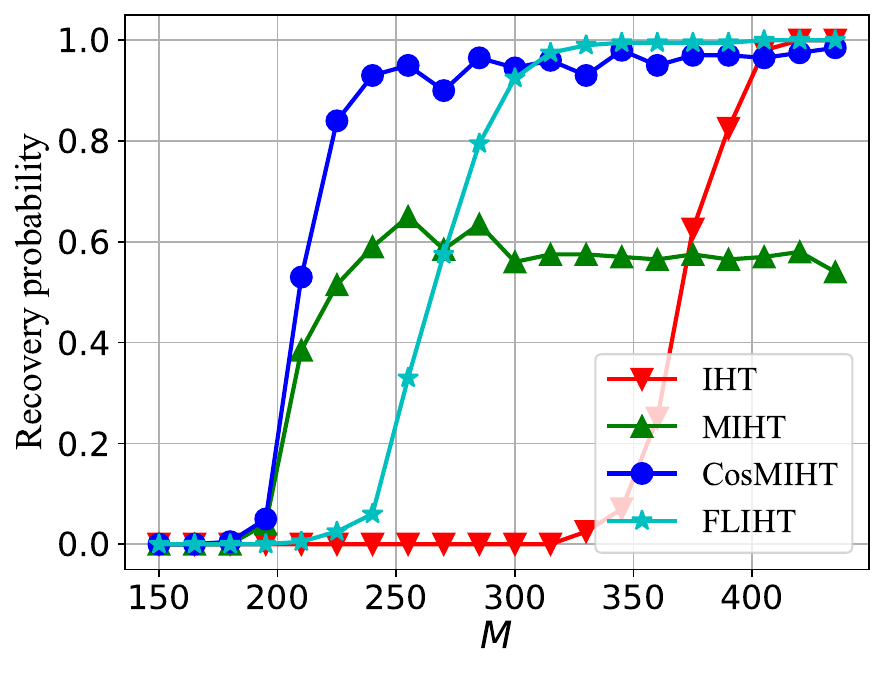}
        \subcaption{$\sigma=10^{-3}$}
    \end{subfigure}
    \caption{Recovery Probability vs $M$ for noiseless and noisy measurements.}
    \label{fig:psucc-vs-m}
\end{figure}
\begin{figure}[t!]
    \centering
    \begin{subfigure}{0.24\textwidth}
    \centering
        \includegraphics[scale=0.3]{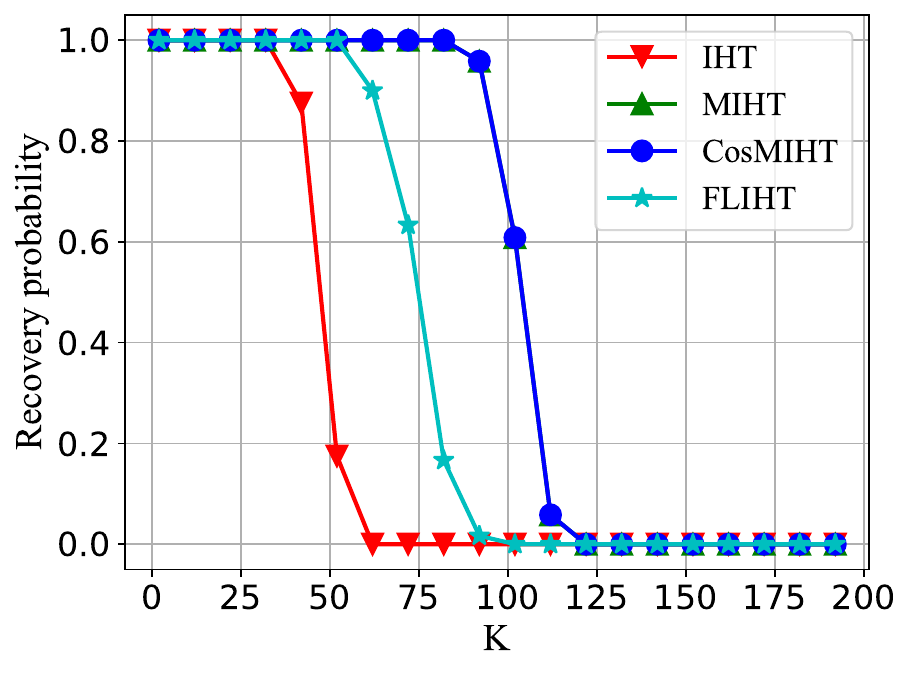}
        \subcaption{$\sigma=0.$}
    \end{subfigure}
    \begin{subfigure}{0.24\textwidth}
    \centering
        \includegraphics[scale=0.3]{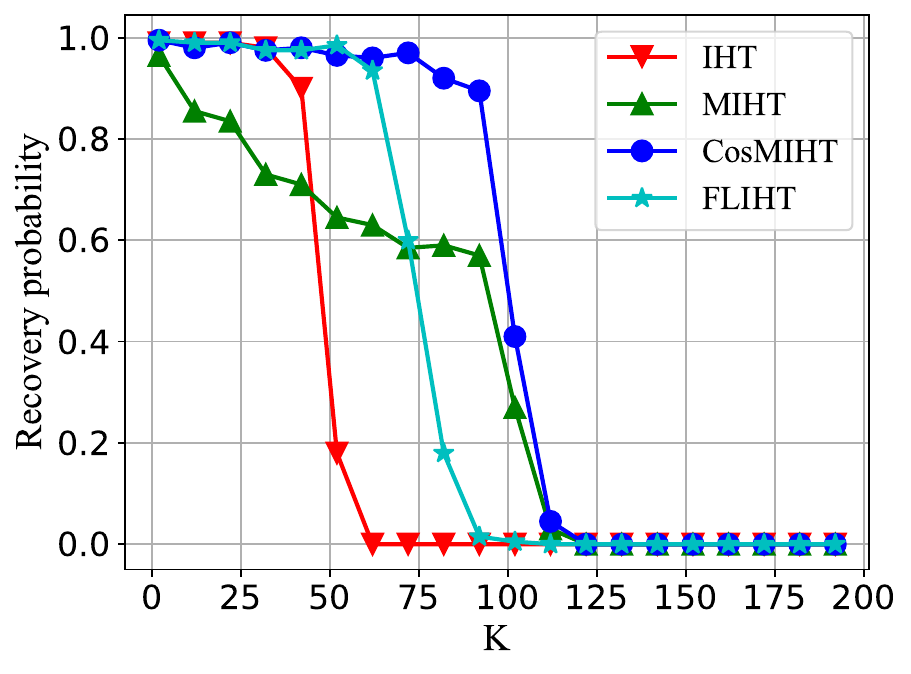}
        \subcaption{$\sigma=10^{-3}$}
    \end{subfigure}
    \caption{Recovery Probability vs $K$ for noiseless and noisy measurements.}
    \label{fig:psucc-vs-k}
\end{figure}
\begin{figure}[t!]
    \centering
    \begin{subfigure}{0.24\textwidth}
    \centering
        \includegraphics[scale=0.3]{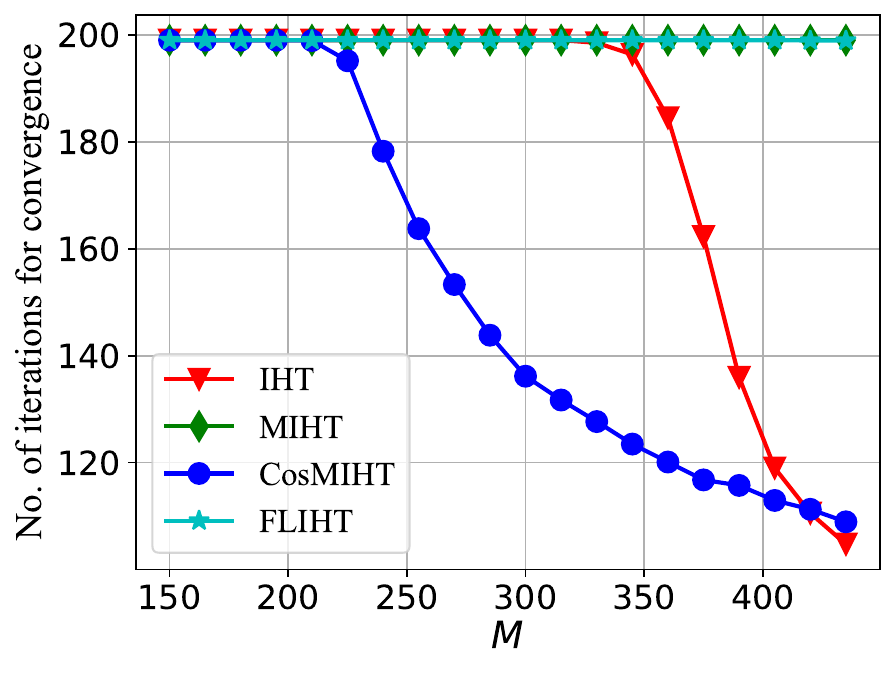}
        \subcaption{$\sigma=0.$}
    \end{subfigure}
    \begin{subfigure}{0.24\textwidth}
    \centering
        \includegraphics[scale=0.3]{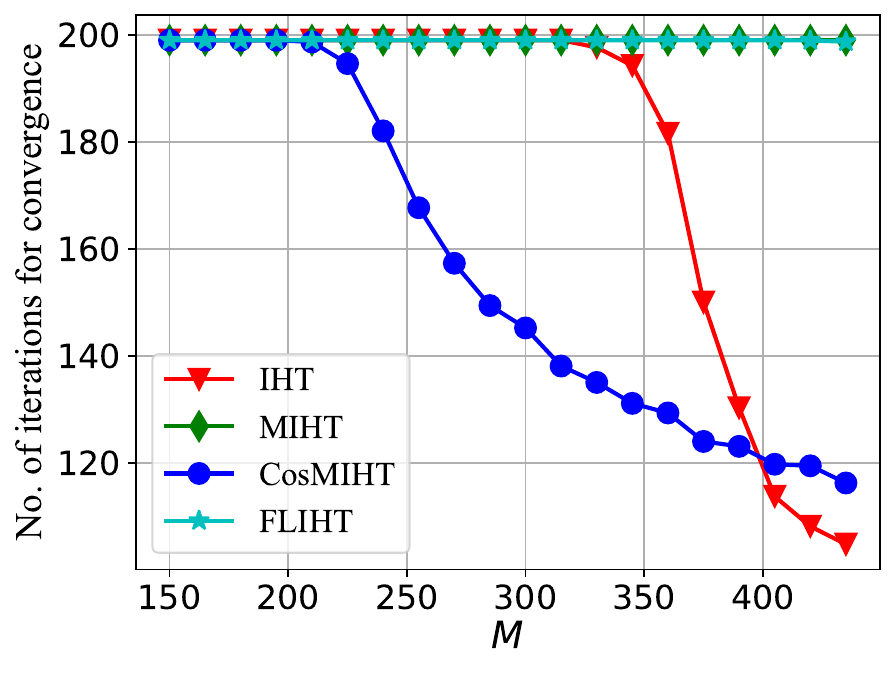}
        \subcaption{$\sigma=10^{-3}$}
    \end{subfigure}
    \caption{No. of iterations for convergence vs $M$ for noiseless and noisy measurements.}
    \label{fig:no-it-vs-m}
\end{figure}
\begin{figure}[t!]
    \centering
    \begin{subfigure}{0.24\textwidth}
    \centering
        \includegraphics[scale=0.3]{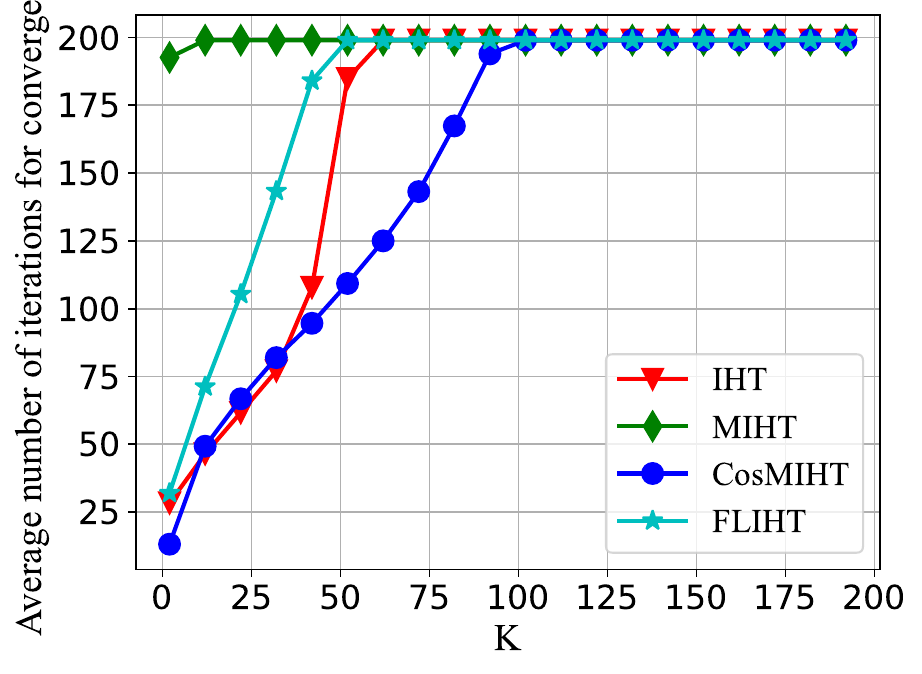}
        \subcaption{$\sigma=0.$}
    \end{subfigure}
    \begin{subfigure}{0.24\textwidth}
    \centering
        \includegraphics[scale=0.3]{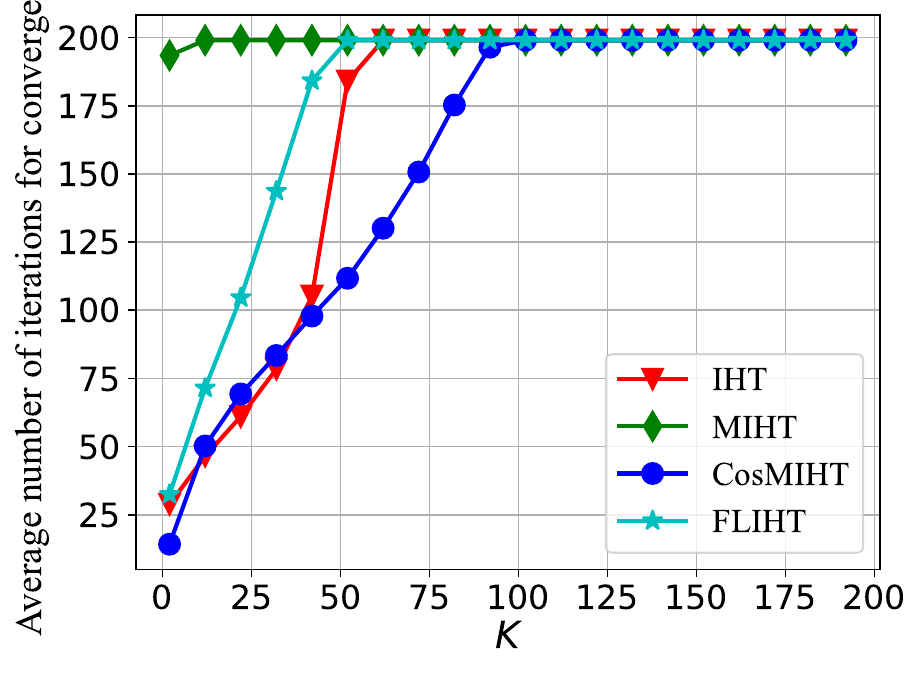}
        \subcaption{$\sigma=10^{-3}$}
    \end{subfigure}
    \caption{No. of iterations for convergence vs $K$ for noiseless and noisy measurements.}
    \label{fig:no-it-vs-k}
\end{figure}
In this section we perform extensive numerical simulations to test the recovery performance as well as convergence speed of the proposed algorithms and compare them with the state-of-the-art methods, for which we have chosen the MIHT as the benchmark, as well as the classical IHT~\cite{BlumensathDavies2009} and the FLIHT~\cite{Cevher2011}. We use Monte-Carlo simulations where we generate $200$ independent instances of the matrix $\bA$ as well as the unknown vector $\bx^\star$ and the noise $\be$, execute the recovery algorithms for each of these instances and compute the recovery metrics for each of the algorithms and finally obtain the recovery performance of the algorithms by averaging the metrics over all these instances. In each of these instances, we generate a random $\R^{M\times N}$ Gaussian matrix with i.i.d.$\sim\mathcal{N}(0,1)$ entries and then $l_2-$normalize the columns to obtain $\bA$. Furthermore, we randomly select a set of $K$ elements as the support of $\bx^\star$ and then generate $K$ i.i.d.$\sim\mathcal{N}(0,1)$ numbers as the entries of $\bx^\star$. Moreover, we generate $\be\sim\mathcal{N}(\bm{0},\sigma\bI_M)$, where $M\sigma^2$ is the noise power. For the simulation setup, we consider $N=1000$, the range of $M$ is taken to be $150$ to $450$ (in steps of $20$) and the range of $K$ is taken to be $2$ to $200$ (in steps of $10$). We use the parameters $\mu=0.6, \beta=0.8$ for MIHT (as guided by~\cite{JinXie2025}), $\mu=0.6$ for IHT, and for FLIHT, we use $\mu=\frac{5}{(1+\sqrt{6K/M}+\sqrt{2})^2}$, which is very close to $5/(1+\delta_{6K})$ with probability larger than $1-e^{-M}$ (see~\cite{Cevher2011}) and performs well empirically.  The experiments have been performed on an Apple Macbook Pro 16 core Laptop. 
In Figure~\ref{fig:error-vs-it}, the mean squared error $\norm{\bx^t-\widehat{\bx}}/\norm{\widehat{\bx}}$ is plotted with respect to iterations for $\sigma=0$ and $\sigma=10^{-3}$. For the noiseless case, the two-phase convergence behavior of \algoname\ is clearly observed, whereas for the noisy case the behavior is less pronounced, as the MSE is to be expected to be dominated by the statistical error induced by measurement noise. 

We evaluate the (empirical) probability of successful recovery as the fraction of instances for which $\frac{\norm{\bx^T-\widehat{\bx}}}{\norm{\widehat{\bx}}}<10^{-3}$. In Figs.~\ref{fig:psucc-vs-m} and~\ref{fig:psucc-vs-k}, the probability of success are plotted with respect to $M$ (with fixed $K=80$) and $K$ (with fixed $M=256$), respectively for the cases $\sigma=0$ and $\sigma=10^{-3}$. In the noiseless cases (Figs.~\ref{fig:psucc-vs-m}(a) and~\ref{fig:psucc-vs-k}(a)) it is clearly observed that the probability of recovery performance of \algoname is exactly the same as that of MIHT and surpasses those of vanilla IHT and FLIHT. This is expected as the support searching phase of \algoname\ is very similar to that of MIHT with the post lock-in behavior being markedly different. This effect is pronounced for the noisy case ($\sigma=10^{-3}$) in Figs.~\ref{fig:psucc-vs-m}(b) and ~\ref{fig:psucc-vs-k}(b), where the probability of recovery of \algoname\ is significantly higher than that of MIHT and having the best critical values, while being comparable to those of IHT and FLIHT.

We compute the number of iterations required for convergence to be the first time $\tau$ such that $\norm{\bx^{\tau+1}-\bx^{\tau}}<10^{-10}$. We plot the mean number of iterations for convergence required for the different algorithms for noiseless and noisy measurements in Figs.~\ref{fig:no-it-vs-m} and~\ref{fig:no-it-vs-k}. In all the Figs. it is clear that \algoname\ enjoys significantly reduced number of iterations for convergence compared to all the other algorithms in the range of comparisons. Only IHT is comparable or maybe even can have smaller number of iterations for convergence, however at the expense of inferior probability of success, as can be seen from the corresponding ranges in Figs.~\ref{fig:psucc-vs-m} and~\ref{fig:psucc-vs-k}.  
%
%


\bibliography{biblio}
\bibliographystyle{IEEEtran}
\appendices
\section{}
\label{appendix:proof-prop2}
We first note that $\by = \bA\bx^\star+\be = \bA\widehat{\bx}+\dualproj{S^\star}\be$. Therefore, it follows directly from the analysis of \texttt{MIHT} in Theorem 1 of~\cite{JinXie2025} that the following holds for any time $t$:
    \begin{align}
    \label{eq:CoSMIHT-one-step-inequality}
        \begin{bmatrix}
            E_{t} \\
            V_{t}
        \end{bmatrix} & \preceq \bQ_t \begin{bmatrix}
            E_{t-1} \\
            V_{t-1}
        \end{bmatrix} + \xi_t \norm{\dualproj{S^\star}\be}\begin{bmatrix}
           1\\
           1
        \end{bmatrix},
    \end{align}
    where 
    \begin{align}
        \bQ_t & = \begin{bmatrix}
            \omega \rho_t & \omega {\beta}_t\\
            {\mu}_t\sqrt{1+\delta_{2K}} & {\beta}_t
        \end{bmatrix},\nonumber\\
        \xi_t & = \omega\mu_t\sqrt{1+\delta_{3K}},
    \end{align}
    where $\rho_t = \abs{\mu_t-1}+\mu_t\delta_{3K}.$ Now, let $t_0=0$ and define, for any $k\ge 0$
    \begin{align}
        \tau_k & = \min\{t> t_{k}:S_{t}\ne S_{t-1}\},\\
        t_{k+1} & = \min\{t> \tau_{k}:S_t=S_{t-1}\}.
    \end{align}
    Then it is clear that $\tau_0=1$ and for any $k\ge 0$, $S_t\ne S_{t-1}$ for $\tau_k\le t\le t_{k+1}-1$ and $S_t=S_{t-1}$ for any $t_k\le t\le \tau_{k}-1$. Therefore, from the description of Algorithm~\ref{alg:CoSMIHT} it follows that
    \begin{align}
        \begin{bmatrix}
            \mu_t\\
            \beta_t
        \end{bmatrix} & = 
        \begin{cases}
            \begin{bmatrix}
                \mu \\
                \beta
            \end{bmatrix}, & 1\le t\le t_1-1\\
            \begin{bmatrix}
                \frac{4}{\left(\sqrt{L_{t_k}}+\sqrt{l_{t_k}}\right)^2}\\
                \left(\frac{\sqrt{L_{t_k}}-\sqrt{l_{t_k}}}{\sqrt{L_{t_k}}+\sqrt{l_{t_k}}}\right)^2
            \end{bmatrix},
             & t_k\le t\le t_{k+1}-1,\ k\ge 1.
        \end{cases}
    \end{align}
    Now note that, for any $t$,
    \begin{align}
        \lefteqn{\frac{4}{(\sqrt{L_t}+\sqrt{l_t})^2} \stackrel{{\color{blue}\textrm{AM-GM}}}{\le} \frac{1}{\sqrt{L_tl_t}}\stackrel{{\color{blue} \text{Eq.}~\eqref{eq:L-l-bounds}}}{\le} \frac{1}{l_t}\stackrel{{\color{blue}\text{Eq.}~\eqref{eq:l-upper-bound}}}{\le} \frac{1}{1-\delta_K}\le\frac{1}{1-\delta_{3K}},} & & \nonumber\\
        \ & \left(\frac{\sqrt{L_{t_k}}-\sqrt{l_{t_k}}}{\sqrt{L_{t_k}}+\sqrt{l_{t_k}}}\right)^2\stackrel{{\color{blue}\text{Eq.}~\eqref{eq:L-upper-bound}}}{\le }  \left(\frac{\sqrt{1+\delta_K}-\sqrt{1-\delta_K}}{\sqrt{1+\delta_K}+\sqrt{1-\delta_K}}\right)^2\nonumber\\
        \ & = \left(\frac{1-\sqrt{1-\delta_K^2}}{\delta_K}\right)^2\le \delta_K^2\le \delta_{3K}^2.
    \end{align}
    Consequently, we get 
    \begin{align}
        \mu_t & \le \max\left\{\mu,\frac{1}{1-\delta_{3K}}\right\}=:\gamma,\nonumber\\
        \beta_t & \le \max\left\{\beta, \delta_{3K}^2\right\}=:\alpha,\nonumber\\
        \rho_t & \le \max\left\{\abs{\mu-1}+\mu\delta_{3K},\frac{2\delta_{3K}}{1-\delta_{3K}}\right\}=:\rho.
    \end{align}
    Using this upper bound in the evolution inequality~\eqref{eq:CoSMIHT-one-step-inequality} proves the assertion.
    \section{}
    \label{appendix:proof-prop3}
     We obtain from~\eqref{eq:CoSMIHT-evolution-base} that 
    \begin{align}
        \begin{bmatrix}
           E_t \\
           V_t
        \end{bmatrix} & \le \bm{Q}^t\begin{bmatrix}
            E_0\\
            V_0
        \end{bmatrix}  + \xi\norm{\be}\sum_{j=0}^{t-1}\bQ^j\begin{bmatrix}
            1 \\ 
            1
        \end{bmatrix}.
    \end{align}
    Note that as all the entries of $\bQ$ are positive, the Perron-Frobenius theorem ensures that there are two unequal real valued eigenvalues of $\bQ,$ denoted by $\lambda_1,\ \lambda_2$ such that $\lambda_1>\abs{\lambda_2}$. Furthermore, if $\bv_1,\bv_2$ are two eigenvectors corresponding to $\lambda_1,\lambda_2$, respectively, then $\bv_1$ has positive entries and the entries of $\bv_2$ have opposite signs. 
    Since the eigenvalues are distinct, one can diagonalize $\bQ$ as $\bQ = \bV\bLambda\bV^{-1}$, where $\bLambda=\begin{bmatrix}
        \lambda_1 & 0\\
        0 & \lambda_{2}
    \end{bmatrix}$, and $\bV = \begin{bmatrix}
        \bv_1 & \bv_{2}
    \end{bmatrix}$, where $\bv_{1},\bv_2$ are eigenvectors corresponding to the eigenvalues $\lambda_1,\lambda_2$, respectively. Also, denoting $\bV^{-1}=\begin{bmatrix}
        \bu^\top_1\\
        \bu_2^\top
    \end{bmatrix}$, one finds that for any $t\ge 0$,
    \begin{align}
        \bQ^t & = \lambda_1^t\bv_1\bu_1^\top+\lambda_2^t\bv_2\bu_2^\top.
    \end{align}
    Therefore, from~\eqref{eq:CoSMIHT-one-step-inequality}, using $E_0=\norm{\bx^0-\widehat{\bx}}=\norm{\widehat{\bx}}$ and $V^0=\norm{\mathcal{H}_{2K}(\bm{0})}=0$, one obtains, 
    \begin{align}
        \lefteqn{\begin{bmatrix}
           E_t \\
           V_t
        \end{bmatrix} \preceq\norm{\widehat{\bx}}(\lambda_1^t\bv_1u_{11}+\lambda_2^t\bv_2u_{21})} & &\nonumber\\
        \ & +\xi\norm{\be}\left(\frac{1-\lambda_1^{t}}{1-\lambda_1}\bv_1(u_{11}+u_{12})+\frac{1-\lambda_2^{t}}{1-\lambda_2}\bv_2(u_{21}+u_{22})\right).
    \end{align}
    Recalling the expression of $\bQ$ from~\eqref{eq:P-mat-expression}, denoting $b=\sqrt{\omega\alpha\gamma\sqrt{1+\delta_{3K}}},$ $\tan \theta = \frac{\alpha-\omega\rho}{2b}$, and setting $\bv_i=\begin{bmatrix}
        1\\
        v_{i2}
    \end{bmatrix}, i=1,2$, it is straightforward to compute the following for $i=1,2$,
    \begin{align}
       \ &  \begin{aligned}
            \lambda_i = \omega\rho + b(\tan\theta+(-1)^{i-1}\sec\theta),
        \end{aligned}\\
        \ & \begin{aligned}
        \bv_{i} = \frac{b}{\omega\alpha\cos\theta}\begin{bmatrix}
            \frac{\omega\alpha\cos\theta}{b}\\
            \sin\theta+(-1)^{i-1}
        \end{bmatrix}, 
        \end{aligned}\\
        \ & \begin{aligned}
            \bu_i & = \frac{(-1)^{i}}{2}\begin{bmatrix}
                \sin\theta+(-1)^{i}\\
                -\frac{\omega\alpha\cos\theta}{b}
            \end{bmatrix}.
        \end{aligned}
    \end{align}
    To find a sufficient condition ensuring $\lambda_1<1$, note that 
    \begin{align}
    \ & \lambda_1 < 1 \Leftrightarrow \frac{\omega\rho+\alpha+\sqrt{(\omega\rho-\alpha)^2+4\omega\alpha\gamma\sqrt{1+\delta_{3K}}}}{2} < 1\nonumber\\
    \ & \Leftrightarrow (\omega\rho-\alpha)^2+4\omega\alpha\gamma\sqrt{1+\delta_{3K}} < (\omega\rho+\alpha-2)^2\nonumber\\
    \ & \Leftrightarrow \omega\alpha\gamma\sqrt{1+\delta_{3K}}<\omega\rho\alpha - (\omega\rho+\alpha)+1\nonumber\\
    \ & \Leftrightarrow  \omega\alpha\gamma\sqrt{1+\delta_{3K}}< (1-\omega\rho)(1-\alpha),\nonumber\\
    \ & \Leftrightarrow \omega\rho < 1, 0<\alpha < 1, \frac{\omega\gamma\sqrt{1+\delta_{3K}}}{1-\omega\rho} < \frac{1-\alpha}{\alpha}.
    \end{align}
    Note that one necessary condition for the above to hold is $\rho<\frac{1}{\omega}$. Recalling the expression for $\rho$ from~\eqref{eq:rho-xi}, it follows that the following is necessary to get $\lambda_1<1$: 
    \begin{align}
        \frac{2\delta_{3K}}{1-\delta_{3K}} < \frac{1}{\omega} \Leftrightarrow \delta_{3K}<\frac{1}{2\omega+1}.
    \end{align}
    Now, we claim that if $\mu<\frac{1}{1-\delta_{3K}}$ (so that $ \gamma = 1/(1-\delta_{3K})$), $\lambda_2<0$. To see this, first note that,
    \begin{align}
        \lambda_1\lambda_2 & = \omega\alpha\left(\rho - \gamma\sqrt{1+\delta_{3K}}\right)\nonumber\\
        \label{eq:eig-product-intermediate}
        \ & = \frac{\omega\alpha}{1-\delta_{3K}} \left(\rho(1-\delta_{3K}) - \sqrt{1+\delta_{3K}}\right).
    \end{align}
    Now we consider two cases.
    \paragraph{$0<\mu\le \frac{(1-3\delta_{3K})}{(1-\delta_{3K})^2}$} 
    In this case, using $\delta_{3K}<\frac{1}{2\omega+1}$ it is straightforward to verify that $\mu<1 $ and that $\rho = 1-\mu(1-\delta)$, so that,~\eqref{eq:eig-product-intermediate} yields,
    \begin{align}
        \lambda_1\lambda_2 & = \frac{\omega\alpha}{1-\delta_{3K}}((1-\mu(1-\delta_{3K}))(1-\delta_{3K})-\sqrt{1+\delta_{3K}})<0.
    \end{align}
    Therefore, as $\lambda_1>0$, in this case, $\lambda_2<0$.
    \paragraph{ $\frac{(1-3\delta_{3K})}{(1-\delta_{3K})^2}\le \mu\le \frac{1}{1-\delta_{3K}}$} In this case, it is easy to verify that $\rho = \frac{2\delta_{3K}}{1-\delta_{3K}}$ so that from~\eqref{eq:eig-product-intermediate}
    \begin{align}
        \lambda_1\lambda_2 & = \frac{\omega\alpha}{1-\delta_{3K}} (2\delta_{3K}-\sqrt{1+\delta_{3K}})).
    \end{align}
    As $\delta_{3K}<\frac{1}{2\omega+1}=\frac{1}{\sqrt{5}+2}$, it is easy to verify that the right-hand side (RHS) above is negative. Since $\lambda_1>0$, this implies that $\lambda_2<0$.

    Since the spectral radius is $\lambda_1$, if $\lambda_1<1$, we also have $\abs{\lambda_2}<1$, so that the following holds
    \begin{align}
        \lefteqn{\begin{bmatrix}
           E_t \\
           V_t
        \end{bmatrix} \preceq\norm{\widehat{\bx}}(\lambda_1^t\bv_1u_{11}+\lambda_2^t\bv_2u_{21})} & &\nonumber\\
        \ & +\xi\norm{\be}\left(\frac{\bv_1(u_{11}+u_{12})}{1-\lambda_1}+\frac{\bv_2(u_{21}+u_{22})}{1-\lambda_2}\right).
    \end{align}
    After some straightforward but tedious algebraic calculations, one can derive the inequalities~\eqref{eq:Et-wandering-equation} and~\eqref{eq:Vt-wandering-equation} from the above. 
    \section{}
    \label{appendix:proof-prop4}
    To proceed further, let us first denote $\bu^t = \begin{bmatrix}
    \bh^{t}\\
    \bh^{t-1}
\end{bmatrix}$, so that the post lock-in evolution of \algoname\ produces the following:
\begin{align}
    \bu^{t+1} & =\begin{bmatrix}
        (1+\betahat)\bI_K - \muhat\bPhi & -\betahat\bI_K\\
        \bI_K & \bm{0}
    \end{bmatrix}\bu^t,\ t\ge \tlock+1.
\end{align}
Further, let the eigenvalues of $\bPhi$ be denoted by $1+\delta_K\ge \lambda_1\ge\cdots\lambda_K\ge 1-\delta_K$, and $\bPhi = \bU^\top\bLambda\bU$ be the corresponding eigenvalue decomposition. Furthermore, let $\bPi$ be a $2K\times 2K$ permutation matrix with the following elements:
\begin{align}
    \Pi_{ij} & = \begin{cases}
        1, & \mbox{$i$ is odd,\ $j=\frac{i+1}{2}$}\\
        1, & \mbox{$i$ is even,\ $j=K+\frac{i}{2}$}\\
        0. & \mbox{else}
    \end{cases}
\end{align}
Then, with $\bDelta^t = \bPi\begin{bmatrix}
    \bU & \bm{0}\\
    \bm{0} & \bm{U}
\end{bmatrix}\bu^t,\ t\ge \tlock$, one can show the following:
\begin{align}
    \bDelta^{t+1} & = \begin{bmatrix}
        \bC_1 & \bm{0} & \cdots & \bm{0}\\
        \bm{0} & \bC_2 & \cdots & \bm{0}\\
        \vdots & \vdots & \ddots & \vdots \\
        \bm{0} & \bm{0} & \cdots & \bC_K
    \end{bmatrix}\bDelta^t,
\end{align}
where \begin{align}
    \bC_i = \begin{bmatrix}
        1+\betahat-\muhat\lambda_i & -\betahat\\
        1 & 0
    \end{bmatrix}, i=1,\cdots, K.
\end{align}
To continue further, let $\ba^{t}_i=\begin{bmatrix}
    \Delta^t_{2i-1}\\
    \Delta^t_{2i}
\end{bmatrix}$ for $i=1,\cdots, K$, $t\ge \tlock+1$. Then, we have, 
\begin{align}
    \ba_i^{t+1} & = \bC_i\ba_i^t,\ i=1,\cdots, K.
\end{align}
Let us denote the eigenvalues of $\bC_i$ by $\nu_{ji},j=1,2$. Then depending on whether or not $\nu_{1i}=\nu_{2i}$, we have two cases:
\paragraph{$\nu_{1i}=\nu_{2i}=\nu_i$ (critically damped case)}
In this case, using Jordan decomposition, one can obtain vectors $\bc_1,\bc_2$ and $\bd_1,\bd_2$ such that for any $\tau\ge 1$,
\begin{align}
    \bC_i^\tau = \nu_i^{\tau}(\bc_1\bd_1^\top+\bc_2\bd_2^\top+\tau\nu_i^{-1}),
\end{align}
so that 
\begin{align}
    \ & \ba_i^{t} = \bC_i^{t-\tlock} \ba_i^{\tlock}\nonumber\\
    \ & \implies \norm{\ba_i^{t+1}}\le \abs{\nu_i}^{t-\tlock}(\abs{\nu_i}\norm{\bc_1}\abs{\bd_1^\top\ba_i^{\tlock}}\nonumber\\
    \ & +\abs{\nu_i}\norm{\bc_2}\abs{\bd_2^\top\ba_i^{\tlock}}+(t-\tlock))\nonumber\\
    \ & \implies \norm{\ba_i^{t+1}} =\mathcal{O}(\abs{\nu_i}^{t-\tlock}).
\end{align}
Note that for the optimal heavy-ball parameters, this critically damped case is the typical regime, and the polynomial prefactor in $t$ is at most linear, which is absorbed by the exponential decay.
\paragraph{$\nu_{1i}\ne \nu_{2i}$ (over/under-damped case)} In this case, using similarity transformation one can diagonalize $\bC_i$ and one can obtain the following for some vectors $\bc_1,\bc_2$ and $\bd_1,\bd_2$ and for any $\tau\ge 1$:
\begin{align}
    \bC_i^\tau = \nu_{1i}^{\tau}\bc_1\bd_1^\top+\nu_{2i}^{\tau}\bc_2\bd_2^\top,
\end{align}
so that 
\begin{align}
    \ & \ba_i^{t} = \bC_i^{t-\tlock} \ba_i^{\tlock}\nonumber\\
    \ & \implies \norm{\ba_i^{t+1}}\le \abs{\nu_{1i}}^{t+1-\tlock}\norm{\bc_1}\abs{\bd_1^\top\ba_i^{\tlock}}\nonumber\\
    \ & +\abs{\nu_{2i}}^{t+1-\tlock}\norm{\bc_2}\abs{\bd_2^\top\ba_i^{\tlock}}\nonumber\\
    \ & \implies \norm{\ba_i^{t+1}} =\mathcal{O}(\max_{j}\abs{\nu_{ji}}^{t+1-\tlock}).
\end{align}
Therefore, in any case, if $\max_{j}\abs{\nu_{ji}}<1$, one can write, 
\begin{align}
    \norm{\ba_i^{t+1}} & \le \mathcal{O}(\max_{j}\abs{\nu_{ji}}^{t-\tlock})
\end{align}
Consequently, if $\max_{ij}\abs{\nu_{ji}}<1$, then one can write, for $t\ge \tlock+1,$ 
\begin{align}
    \norm{\bDelta^{t+1}} & \le \mathcal{O}({\sqrt{K}\max_{ij}\abs{\nu_{ji}}^{t-\tlock}}).
\end{align}
Since $\norm{\bh^{t+1}}\le \norm{\bu^{t+1}}=\norm{\bDelta^{t+1}}$, we get for $t\ge \tlock+1.$
\begin{align}
    E_{t}=\norm{\bx^{t}-\xhat} & \le \mathcal{O}({\sqrt{K}\max_{ij}\abs{\nu_{ji}}^{t-1-\tlock}}).
\end{align}
Therefore, $E_t\to 0$ as $t\to\infty$ i.e, $\bx^t\to \xhat$ as $t\to\infty$ as long as the spectral radius satisfies $\max_{ij}\abs{\nu_{ji}}<1$. One can easily check that the eigenvalues of $\bC_i$ are obtained as roots of the following equation:
\begin{align}
\label{eq:eigs}
    \nu^2-(1+\betahat-\muhat\lambda_i)\nu+\betahat=0,
\end{align}
which has solutions 
\begin{align}
\label{eq:nu1}
    \nu_{1i} & = \frac{1+\betahat-\muhat\lambda_i+\sqrt{(1+\betahat-\muhat\lambda_i)^2-4\betahat}}{2},\\
    \label{eq:nu2}
    \nu_{2i} & = \frac{1+\betahat-\muhat\lambda_i-\sqrt{(1+\betahat-\muhat\lambda_i)^2-4\betahat}}{2},
\end{align}
We now obtain conditions which ensure that, for a fixed $i$, the spectral radius is less than $1$ by analyzing the sign of $(1+\betahat-\muhat\lambda_i)^2-4\betahat$.
\begin{enumerate}
    \item $(1+\betahat-\muhat\lambda_i)^2\le 4\betahat$: In this case, it is easy to see that $\sqrt{\muhat\lambda_i}\in \left[1-\sqrt{\betahat}, 1+\sqrt{\betahat}\right] $ and $\abs{\nu_{1i}} = \abs{\nu_{2i}}=\sqrt{\betahat}$. Hence the spectral radius is $\sqrt{\betahat}$.
    \item $(1+\betahat-\muhat\lambda_i)^2> 4\betahat$: In this case, $\sqrt{\muhat\lambda_i}\notin \left[1-\sqrt{\betahat}, 1+\sqrt{\betahat}\right]$ and the spectral radius can be found to be $\frac{\abs{1+\betahat-\muhat\lambda_i}+\sqrt{(1+\betahat-\muhat\lambda_i)^2-4\betahat}}{2}>\sqrt{\betahat}.$
\end{enumerate}
Therefore, one can write
\begin{align}
    \max_{\substack{1\le i\le K\\ j\in{1,2}}}\abs{\nu_{ji}} \le \sqrt{\betahat}\,g\left(\max\left\{\max_{1\le i\le K}\abs{\frac{1+\betahat-\muhat\lambda_i}{2\sqrt{\betahat}}},1\right\}\right),
\end{align}
where $g(x) = x+\sqrt{x^2-1}$.
\section{}
\label{appendix-proof-prop5}
\textbf{Analysis for $L$}:
 Let the spectral decomposition of $\bB$ be given by $\bB = \bR\bLambda\bR^\top$, where $\bR=[\br_1\cdots,\br_K]$ is a unitary matrix and $\bLambda = \mathrm{diag}(\lambda_1,\cdots, \lambda_K)$. Using the spectral decomposition of $\bB$ as $\bB = \sum_{i=1}^K \lambda_i\br_i\br_i^\top$, we obtain,
\begin{align}
    L & = \frac{\sum_{i=1}^K \lambda_i^{2r+1}w_i}{\sum_{i=1}^K \lambda_i^{2r}w_i} = \lambda_1\cdot\left(\frac{w_1+\sum_{i=2}^K \mu_i^{2r+1}w_i}{w_1+\sum_{i=2}^K \mu_i^{2r}w_i}\right),
\end{align}
where $w_i=\frac{(\bg^\top\br_i)^2}{\norm{\bg}^2},\ 1\le i\le K$, and $\mu_i = \frac{\lambda_i}{\lambda_1},\ 2\le i\le K$. Therefore, we find that,
\begin{align}
    0\le 1-\frac{L}{\lambda_1} & = \frac{\sum_{i=2}^K(1-\mu_i)\mu_i^{2r}w_i}{w_1+\sum_{i=2}^K\mu_i^{2r}w_i}\nonumber\\
    \ & \le \frac{(1-\mu_K)\mu_2^{2r}(\sum_{j\ne 1}w_j)}{w_1}.
\end{align}
Letting $w_1=\cos^2\psi$, so that $\sum_{j\ne 1}w_j=\norm{\br_1}^2-w_1=\sin^2\psi$, one obtains,
\begin{align}
     0\le 1-\frac{L}{\lambda_1} & \le (1-\mu_K)\mu_2^{2r}\tan^2\psi\nonumber\\
     \label{eq:L-upper-bound}
     \implies \lambda_1\ge L & \ge \lambda_1(1-(1-\kappa^{-1})\mu_2^{2r}\tan^2\psi),
\end{align}
where we recognized that $\mu_K=\frac{\lambda_K}{\lambda_1}=\kappa^{-1}$, where $\kappa$ is the condition number of $\bB$.
\ \\
\textbf{Analysis for $l$:}
Let us denote $\gamma_i = \trace{\bB} - \lambda_i$ and $\eta_i = \frac{\gamma_i}{\gamma_K}$, $1\le i\le K$ . Then it follows, using step similar to the analysis of $L$ above, that 
\begin{align}
    l & = \lambda_K + \gamma_K-\frac{\sum_{i=1}^K \gamma_i^{2r+1}w_i}{\sum_{i=1}^K \gamma_i^{2r}w_i}\nonumber\\
    \implies \frac{l-\lambda_K}{\gamma_K} & = 
    \frac{\sum_{i=1}^{K-1} (1-\eta_i)\eta_i^{2r}w_i}{\sum_{i=1}^{K-1} \eta_i^{2r}w_i+w_K}\nonumber\\
    \ & \le \frac{(1-\eta_1)\eta^{2r}_{K-1}\sum_{i=1}^{K-1}w_i}{w_K}\nonumber\\
    \ & = (1-\eta_1)\eta^{2r}_{K-1}\tan^2\zeta\nonumber\\
    \label{eq:l-upper-bound}
    \implies \lambda_K\le l & \le \lambda_K + (K-1)\lambda_1(1-\kappa^{-1})\eta_{K-1}^{2r}\tan^2\zeta,
\end{align}
where $w_K=\cos^2\zeta$, and we have recognized that $\eta_1=\frac{\sum_{i\ne 1}\lambda_i}{\sum_{j\ne K}\lambda_j}\ge \frac{(K-1)\lambda_K}{(K-1)\lambda_1}=\frac{\lambda_K}{\lambda_1}=\kappa^{-1},$ and that $\gamma_K=\sum_{j\ne K}\lambda_j\le (K-1)\lambda_1$.

Now, we let $\epsilon>0$ be such that 
\begin{align}
\label{eq:epsilon-condition1}
    \lambda_1(1-\kappa^{-1})\mu_2^{2r}\tan^2\psi & \le \epsilon,\\
    \label{eq:epsilon-condition2}
    (K-1)\lambda_1(1-\kappa^{-1})\eta_{K-1}^{2r}\tan^2\zeta & \le \epsilon,
\end{align}
which is ensured if 
\begin{align}
\label{eq:r-bound}
    r\ge \left\lceil\frac{\ln\left(\frac{\lambda_1\max\left\{\tan^2\psi,(K-1)\tan^2\zeta\right\}}{\epsilon}\right)}{2\ln\left(\frac{1}{\max\{\mu_2,\eta_{K-1}\}}\right)}\right\rceil.
\end{align}
Consequently, with $r$ satisfying~\eqref{eq:r-bound}, one obtains,
\begin{align}
\label{eq:L-l-bounds}
    \lambda_1-\epsilon \le L\le \lambda_1, & 
    \ \lambda_K \le l\le \lambda_K+\epsilon.
\end{align}
We now further find upper bounds on $\tan^2\psi$ and $\tan^2\zeta$ with high probability using the initialization as depicted in \algoname\ (line 1 of Algorithm~\ref{alg:CoSMIHT}). 

Let $\bg\sim\mathcal{N}(\bm{0},\bI)$ and let $w_i=\frac{(\bg^\top\br_i)^2}{\norm{\bg}^2}$ for $i=1,\cdots, K$. We will show that, for some threshold $\gamma$ to be chosen later, $w_1\ge \gamma$ and $w_K\ge \gamma$ with high probability. 

To do this, we first write,
\begin{align}
    \lefteqn{\prob\big(w_{1}\ge \gamma,\ w_{K}\ge \gamma\big)} & &\nonumber\\
    \ & \ge 1 - \prob\big(w_1<\gamma\big) - \prob\big(w_K<\gamma\big).
\end{align}
Now, note that $w_1$ as well as $w_K$ are distributed identically as $B\left(\frac{1}{2},\frac{K-1}{2}\right)$. Therefore,
\begin{align}
    \lefteqn{\prob\big(w_1<\gamma\big) = \frac{\int_0^\gamma x^{-1/2}(1-x)^{(K-3)/2}dx}{\beta\left(\frac{1}{2},\frac{K-1}{2}\right)}} & & \nonumber\\
    \ & \stackrel{\color{blue}(x\le 1)}{\le} \frac{\int_0^\gamma x^{-1/2}dx}{\frac{\sqrt{\pi}\Gamma\left(\frac{K-1}{2}\right)}{\Gamma\left(\frac{K}{2}\right)}}\stackrel{\color{blue}\textrm{Gautschi's inequality}}{<} \sqrt{\frac{(K-2)\gamma}{\pi}}.
\end{align}
Consequently, we have, 
\begin{align}
   \prob\big(w_{1}\ge \gamma,\ w_{K}\ge \gamma\big) & \ge 1 - 2\sqrt{\frac{(K-2)\gamma}{\pi}}.
\end{align}
Therefore, to ensure $\prob\big(w_{1}\ge \gamma,\ w_{K}\ge \gamma\big)\ge 1-\delta$ for some $\delta\in(0,1)$, it is sufficient to choose $\gamma=\frac{\delta^2\pi}{4(K-2)}$. Consequently, with probability larger than $1-\delta$, we have,
$\max\{\tan^2\psi,(K-1)\tan^2\zeta\}\le (K-1)\left(\frac{4(K-2)}{\delta^2\pi}-1\right)\le \frac{4K^2}{\delta^2\pi}.$ Therefore, we have the desired result.
\section{}
\label{appendix-proof-thm2}
Therefore, we find,
\begin{align}
    \begin{aligned}
    \frac{1+\betahat-\muhat\lambda^\star_1}{2\sqrt{\betahat}} & = \frac{L+l-2\lambda^\star_1}{L-l}\le \frac{2(L-\lambda^\star_1)}{L-l}\le 0,\\
    \frac{1+\betahat-\muhat\lambda^\star_K}{2\sqrt{\betahat}} & = \frac{L+l-2\lambda^\star_K}{L-l}\ge \frac{2(l-\lambda^\star_K)}{L-l}\ge 0.
\end{aligned}
\end{align}
Consequently, it can be verified that,
\begin{align}
    \lefteqn{\max\left\{\max_{1\le i\le K}\abs{\frac{1+\betahat-\muhat\lambda^\star_i}{2\sqrt{\betahat}}},1\right\}} & &\nonumber\\
    \ & = \max\left\{\frac{2\lambda^\star_1-(L+l)}{L-l},\frac{L+l-2\lambda^\star_K}{L-l},1\right\}.
\end{align}
Since $\frac{2\lambda^\star_1-(L+l)}{L-l}-1=\frac{2(\lambda_1^\star-L)}{L-l}>0$, we further simplify that,
\begin{align}
    \lefteqn{\max\left\{\max_{1\le i\le K}\abs{\frac{1+\betahat-\muhat\lambda^\star_i}{2\sqrt{\betahat}}},1\right\}} & &\nonumber\\
    \ & = \max\left\{\frac{2\lambda^\star_1-(L+l)}{L-l},\frac{L+l-2\lambda^\star_K}{L-l}\right\}\nonumber\\
    \ & = \frac{\lambda^\star_1-\lambda^\star_K+\abs{\lambda^\star_1+\lambda^\star_K-(L+l)}}{L-l}\nonumber\\
    \ & \le \frac{\lambda^\star_1-\lambda^\star_K+\epsilon}{\lambda^\star_1-\lambda^\star_K-2\epsilon}.
\end{align}
With $\epsilon=\epsilon_0(\lambda^\star_1-\lambda^\star_K)$, and $\epsilon_0\le \frac{1}{8}$, we obtain, 
\begin{align}
    \max\left\{\max_{1\le i\le K}\abs{\frac{1+\betahat-\muhat\lambda^\star_i}{2\sqrt{\betahat}}},1\right\} & \le 1+\frac{3\epsilon_0}{1-2\epsilon_0}\le 1+4\epsilon_0.
\end{align}
Furthermore, as $L/l\le \lambda^\star_1/\lambda^\star_K\le \frac{1+\delta_{3K}}{1-\delta_{3K}}$, and $\widehat{\beta}$ is an increasing function of $L/l$, we have,
\begin{align}
    \widehat{\beta} \le \left(\frac{\sqrt{\lambda^\star_1}-\sqrt{\lambda^\star_K}}{\sqrt{\lambda^\star_1}+\sqrt{\lambda^\star_K}}\right)^2 & \le \left(\frac{\sqrt{1+\delta_{3K}}-\sqrt{1-\delta_{3K}}}{\sqrt{1+\delta_{3K}}+\sqrt{1-\delta_{3K}}}\right)^2\nonumber\\
    \ & =\frac{\delta^2_{3K}}{\left(1+\sqrt{1-\delta^2_{3K}}\right)^2}.
\end{align}
Consequently, with the choice of $r$ as in Proposition~\ref{prop:power_accuracy}, one gets a convergence rate of \begin{align}
    \frac{\delta_{3K}\sqrt{\left(1+\frac{4\epsilon}{\lambda_1-\lambda_K}\right)}}{1+\sqrt{1-\delta_{3K}^2}} \le \frac{\delta_{3K}\left(1+\frac{2\epsilon}{\lambda_1-\lambda_K}\right)}{1+\sqrt{1-\delta_{3K}^2}}, 
\end{align}
where the last line follows since $\sqrt{1+x}\le  1+\frac{x}{2}$, for any $x\in \real$.
\end{document}